\documentclass[a4paper,11pt]{article}
\usepackage{pos}

\usepackage{latexsym}
\usepackage[all,knot]{xy}
\xyoption{arc}
\usepackage{tikz-cd}
\usepackage{hyperref}
\usepackage{comment}
\usepackage{cancel}

\newtheorem{theorem}{Theorem}[section]
\newtheorem{lemma}[theorem]{Lemma}
\newtheorem{proposition}[theorem]{Proposition}
\usepackage[dvipsnames]{xcolor}
\usepackage{comment}
\usepackage{cleveref}
\usepackage{tikz-feynman}
\usetikzlibrary{matrix,cd,arrows,decorations.pathmorphing,external}
\usepackage{mathtools}
\usepackage[all,knot]{xy}
\xyoption{arc}

\tikzset{
    int/.style={
        line width=1.5pt,
        line cap=round,
        dash pattern=on 0pt off 3pt
    },
    ambi/.style={
        line width=1.2pt,
        line cap=round,
        dash pattern=on 2pt off 3pt
    },
    anti/.style={decorate, draw=black,
        decoration={zigzag,segment length = 1mm, amplitude = 0.7mm}}
}

\definecolor{JanColor}{rgb}{0.07, 0.40, 0.00}

\title{Field theory of $\mathfrak{su}(n)$: the absence of non-zero scatterings}

\author*[a]{Eugenia Boffo}
\author[a]{Ján Pulmann}
\author[b]{\v L{u}bo\v{s} Ravas}

\affiliation[a]{Mathematical Institute, Faculty of Mathematics and Physics, Charles University,\\
  Sokolovsk\'{a} 83, 186 75, Praha, Czechia}

\affiliation[b]{Institute of Inorganic Chemistry, Department of Theoretical Chemistry,\\
Dúbravská cesta 9, 845 36, Bratislava, Slovakia}

\emailAdd{boffo@karlin.mff.cuni.cz}
\emailAdd{jan.pulmann@gmail.com}
\emailAdd{lubos.ravas@savba.sk}

\abstract{We inspect $\mathfrak{su}(n)$ forms, providing greater detail for $n=2,3$, as a toy model for a field theory in finite dimensions and with gauge symmetries. Relying on homological perturbation theory, we show that there are no scattering amplitudes with trivalent tree-level diagrams, except for the interaction vertex, thus extending a known argument of \cite{Cattaneo:2009deo} to arbitrary $n$. In contrast to this, we show how to obtain non-trivial higher products when transferring to a larger space of fields.}

\FullConference{Proceedings of the Corfu Summer Institute 2025 "School and Workshops on Elementary Particle Physics and Gravity" (CORFU2025)\\
27 April - 28 September, 2025\\
Corfu, Greece\\}

\tableofcontents

\begin{document}
\maketitle

\section{Introduction}

Classical and quantum field theories are the core of theoretical physics.  In most cases there is a formulation which is valid independently on whether the field equations are satisfied, the \emph{off-shell formulation} based on Lagrangian or Hamiltonian action functionals. In other circumstances we might just have access to information at the locus of the field equations, i.e.~we know the \emph{on-shell theory}. An example where this occurs is the $N=2$ superconformal field theory in dimension $4$ \cite{Argyres:1995xn} and also its higher dimensional versions in $5$ and $6$ dimensions. In these situations, we only have access to the scattering amplitudes: they are computed from a perturbative expansion around a fixed point of the action functional (if it exists!).  
Several strategies can be adopted to investigate QFTs either at the locus of the fields equations or off-shell.   

Homotopy theory and homological algebra are outstanding tools for a modern approach to QFTs, culminating in the Batalin--Vilkovisky quantization technique \cite{Batalin:1977pb}. The ultimate goal of BV, in the context of field theories whose path integral degenerates because of gauge symmetries, is to replace the classical theory by a perturbatively quantized theory that satisfies a differential equation, called the \emph{quantum master equation}. 
A modern viewpoint on BV is due to Costello and Gwilliam \cite{Costello:2021jvx} and has been further developed by many authors, of which we give our very partial list: \cite{Jurco:2018sby,Jurco:2019yfd,Jurco:2024gct}. Another remarkable article where the reader can find an excellent bridge between the homotopy language and field theories is \cite{Eager:2021wpi}. Perturbative BV relies on cyclic and quantum $L_\infty$ algebras. The notion of homotopy transfer \cite{Vallette-Loday} is often invoked to import the aforementioned structure to a quasi-isomorphic, equivalent, physical system. Although the prime object of interest is the off-shell theory, also the on-shell information can be unveiled: there is an algorithm for computing scattering amplitudes that follows from the homological perturbation lemma \cite{Saemann:Srninotes}. Similarly to the standard treatment of QFTs, diagrams are produced, the so called HPL diagrams. Even though they have a different mathematical interpretation, they are closely related to Feynman diagrams. See for example \cite{Saemann:2020oyz} for matching of the symmetry factors of these diagrams.

In this note, we will be exploring a finite dimensional field theory: forms of $\mathfrak{su}(n)$, with focus on $n=2,3$. In the former case, it corresponds to some kind of Chern--Simons theory, which mathematically describes the moduli space of flat connections. There is vast literature on the subject of Chern--Simons, initiated by \cite{Witten:1988hf}, but this should not prevent us from exploring it further and looking at it from another viewpoint. One of our starting points was \cite{Nguyen:2021rsa}, where $\mathfrak{su}(2)$-forms with values in the fuzzy sphere (truncated spherical harmonics) is analyzed, with the goal of formulating a braided quantum field theory. 
In the case under inspection here (over a group manifold and with an external Lie algebra), our field theory will not have more than the 3-point interaction vertex, because of the triviality of the composition between the propagator and the product. This prevents us from finding non-vanishing Feynman diagrams above 3 points. From a more general perspective, this result is a consequence of the fact that cohomology embeds as a subalgebra to CE complex, \cite{Cattaneo:2009deo}. 
We confirm it by compiling some facts of representation theory of Lie groups and then by working out the diagrams from homological perturbation theory. The quantum effective action on cohomology was computed for $\mathfrak{su}(2)$ case by Cattaneo--Mn\"{e}v \cite{Cattaneo:2009deo}, section 3.4. We extend the classical argument to $\mathfrak{su}(n)$ and also find non-trivial tree-level contributions in the case of a partial homotopy transfer, which in physics terms means we integrate out a smaller number of fields. This work builds upon some preliminary findings contained in \cite{Ravas:thesis}.

\section{Lie algebra forms in field theory}

Classical field theories with gauge symmetries are the cornerstone of theoretical physics. Usually, the space of fields is infinite dimensional. Focusing on exterior forms of a Lie coalgebra (valued in the ground field $\mathbb{C}$ for a start) reduces the dimensionality to be finite. Then, an enlightening way to arrange the field content of the theory is in a cochain complex, as this accounts not only for the gauge symmetries, but also the field equations and the conservation laws. 
Such idea is at the core of the Batalin--Vilkovisky, approached with the paradigm of ``cohomological physics" \cite{Stasheff:1997iz} (see also \cite{Zwiebach:1992ie,Sachs:2019gue} in string theory). 
One places the gauge parameters, fields and their dual fields in different cochain degrees. In the following we shall explore forms of $\mathfrak{su}(n)$ with $n=2,3$ for concreteness, though our main result holds for generic $n \in \mathbb{N}$. 
We first collect some information on the notion of Chevalley--Eilenberg complex with trivial coefficients and its cohomology.

\paragraph{Chevalley-Eilenberg complex} If $\mathfrak{s}$ is a Lie algebra\footnote{Read $\mathfrak{s}$ for source; see Section \ref{ssec:cyclic}.}, one  defines the \emph{Chevalley-Eilenberg complex} with trivial coefficients by
\[ \Lambda^\bullet \mathfrak{s}^* \]
with a degree $1$ differential
\[ \mathrm{d}(e^i) = f_{jk}{}^i e^je^k, \]
where $e^i$ is a basis of $\mathfrak{s}^*$ dual to $e_i\in \mathfrak s$. The structure constants $f_{jk}{}^i$ are defined by $[e_j,e_k]= f_{jk}{}^i e_i$. This complex is a finite-dimensional model for the cohomology for compact Lie groups integrating $\mathfrak{s}$. More precisely, for any Lie group $S$ the complex $ \Lambda^\bullet \mathfrak{s}^*$ is isomorphic to the complex of left-invariant forms on $S$; while for compact and connected $S$ the inclusion of left-invariant forms into all forms induces an isomorphism on cohomology \cite[Theorem~9.1, Theorem~2.3]{ChevalleyEilenberg47}. Because $\Lambda^\bullet\mathfrak{s}^*$ models $\Omega^\bullet(S)$, we will call elements of $\Lambda^\bullet \mathfrak{s}^*$ \emph{Lie algebra forms}.
The Lie algebras we wish to consider are the $\mathfrak{su}(n)$ algebras, which belong to the class of \emph{reductive} Lie algebras.

\paragraph{$\mathfrak{su(2)}$} Our focus is on alternating forms over $\mathbb{C}$ of the (complex form of the) 
semi-simple Lie algebra{\footnote{As a matter of fact, over $\mathbb{C}$ this is $\mathfrak{sl}(2)$. This complexification will help us with the construction of a special deformation retract.} $\mathfrak{su}(2)$ endowed with the Chevalley--Eilenberg differential:
\[
\left(\Lambda^\bullet \mathfrak{su}(2)^*\right) \, , \; \; \mathrm{d} = \epsilon_{ij}{}^k \, e^i e^j \, e^*_k , \qquad e^i \in \mathfrak{su}(2)^* \, .
\]
This is a commutative differential graded algebra (cdga). Note that, as a manifold, $SU(2) \cong \mathbb{S}^3$. We will denote
\[
E= \left(\Lambda^\bullet \mathfrak{su}(2)^*\right)[1] \,.
\]
The degree shift is essential for a Batalin--Vilkovisky interpretation of the complex, as the ghost numbers are assigned according to
\begin{equation}
   \text{gh}(a) = i-1 \,, \quad a \in \Lambda^i\mathfrak{su}(2)^* \; .
\end{equation}
Hence linear functions of $\Lambda^0$, that have ghost degree $+1$, are the \emph{gauge parameters}; linear functions of $\Lambda^1$, whose ghost number is null, are the \emph{gauge fields} and linear functions of $\Lambda^2$, which carry ghost number $-1$, are their \emph{antifields}. Finally, the linear functions of $\Lambda^3$ are the \emph{antifields of the gauge parameters}.   This is typical of Chern--Simons theories \cite{Witten:1988hf}. As appropriate to the BV framework, the $\mathbb{Z}_2$-parity of the gauge parameters and the antifields is odd.

Another ingredient is a pairing on forms: 
\begin{align}
  \langle -,-\rangle: & \,  E \times E \to \mathbb{C}[-1]\, , \notag\\
  &\, a, b \mapsto  \langle a, b \rangle = 
      (-1)^{\vert a \vert} \int a \wedge b \,,  
\end{align}  
In the orthonormal basis of $\mathfrak{su}(2)^*$, the integral picks the coefficient of the top form $e^1 \wedge e^2 \wedge e^3$. Note that because of the $-1$ shift, the dual space is $E^*\equiv E[1]$.

We shall now unveil some information on the space of inequivalent solutions to the field equation $\mathrm{d} a = 0$. The Chevalley-Eilenberg cohomology groups are:
\begin{equation}
H^{i}(E)=  \begin{cases}
     \mathbb{C} , & i=0,3\, ,\\
     \emptyset, &  i=1,2\,,
 \end{cases}   
\end{equation}
by Whitehead lemma or also by homotopy groups of the sphere $\mathbb{S}^3 \cong SU(2)$. 
This highlights, for instance, that the $2$-form, a field strength for the $1$-form, is trivial in cohomology. 

\paragraph{$\mathfrak{su(3)}$} 
The next example to bear in mind is the Chevalley--Eilenberg complex for $\mathfrak{su}(3)$. As a vector space, this is $8$-dimensional. As a group manifold, it is the total space of a fibration 
\[
SU(3) \xrightarrow{\pi} \mathbb{S}^5 \; ,  
\]
with fiber isomorphic to $\mathbb{S}^3$. The Chevalley-Eilenberg cdga is
\[
\Lambda^\bullet \mathfrak{su}(3)^* \, , \; \;\mathrm{d} = f_{ij}{}^k \, e^i e^j \, e^*_k , \qquad e^i \in \mathfrak{su}(3)^* \, .
\]
The structure constants $f_{ij}{}^k$ for $\mathfrak{su}(3)$ are tabled: In the Gell-Mann representation, these are $f_{123}=1\ , \; f_{147}=f_{165}=f_{246}=f_{257}=f_{345}=f_{376}={\frac {1}{2}}\ ,\; f_{458}=f_{678}={\frac {\sqrt {3}}{2}}\ .$ Another convenient choice is the Cartan--Weyl basis of $3$ by $3$ matrices: if $e_{ij}$ stands for the unit in position $i,j$,
\begin{align}
\tilde E_1 = e_{11} -e_{22}, & \quad  \tilde E_2 = e_{22}-e_{33},  \quad \tilde E_3 = e_{12},  \quad \tilde E_4 = - e_{23},\notag \\ 
\tilde E_5 = -e_{13}, & \quad \tilde E_6= e_{21},  \quad \tilde E_7 = -e_{32},  \quad \tilde E_8 = e_{31}\,. \notag
\end{align}
We find it educative to present these bases, though we will not deploy them.

The Batalin--Vilkovisky interpretation is now less clear. Defining the forms over the whole bundle, they can be arranged as:
\begin{equation}
    \Lambda^0 \mathfrak{su}(3)^* \rightarrow \dots \rightarrow \Lambda^3 \mathfrak{su}(3)^* \rightarrow \Lambda^4 \mathfrak{su}(3)^* \rightarrow \Lambda^5 \mathfrak{su}(3)^* \rightarrow \dots \rightarrow \Lambda^8 \mathfrak{su}(3)^*\;,
 \end{equation}
 therefore, a reasonable possibility is to state that lower forms are ghosts for a gauge symmetry and its reducibility, while higher forms are higher N\"other identities. For now, we refrain from shifting the degrees, and we will come back to that later, when we will provide a $2$-product coming from coupling to a different Lie algebra $\mathfrak g$. Another ingredient is a pairing on forms: 
\begin{align}
  \langle -,-\rangle: & \,  E \times E \to \mathbb{C}\, ,\notag\\
  &\, a, b \mapsto \langle a, b \rangle = 
      (-1)^{\vert a \vert}\int a \wedge b \,,   
\end{align} 
that is non-zero when the integrand is a top form (8-form in this case). The normalization is not important for us. The typical shifting of a BV theory will be performed when we provide the product.

Let us now discuss the solutions $\mathrm{d}a =0$ modulo gauge symmetries. A general result about cohomology groups of $SU(n)$ is the following corollary of Hopf--Koszul--Samelson theorem \cite{Hopf}: If the subscript of $x_i$ denotes the degree in $H^i$, the cohomology groups are isomorphic to 
\[
H^\bullet(\mathfrak{su}(n)) \cong \Lambda[x_3, x_5, \dots x_{2n-1}]\;.
\] 
So for $\mathfrak{su}(3)$ the cohomology groups are:
\begin{align}
    H^i(\mathfrak{su}(3)) = \begin{cases}
        \mathbb{C} , \quad i=0,3,5,8,\\
        \emptyset \quad \text{otherwise}.
    \end{cases}
\end{align}
An explicit characterization is possible: the cohomology classes are the invariant exterior forms (it follows from \emph{Hopf-Koszul-Samelson theorem} \cite{Hopf, Meinrenken}) 
\begin{equation}
    H(\mathfrak{s}) \cong \left(\Lambda \mathfrak{s}^*\right)^{\mathfrak{s}}\; .
    \label{HKS}
\end{equation}
Specializing the statement to $\mathfrak{s}=\mathfrak{su}(3)$, the invariant tensor of $\Lambda^0$ is the identity, while that of $\Lambda^3$ is the totally antisymmetrized structure constants $f_{ijk}$.
We note that the invariants of $\Lambda^5$ and $\Lambda^8$ are obtained under Hodge star $\star$ isomorphism  (in fact, $\star (\Lambda^k \mathfrak{s}^*)^\mathfrak{s} = (\Lambda^{d-k}\mathfrak{s})^\mathfrak{s}= (\Lambda^{d-k}\mathfrak{s}^*)^\mathfrak{s}$, the conclusion following by self-duality of the invariants of a reductive Lie algebra). 

\subsection{On-shell fields}
\label{On shell}

In physics, the "shell" is the set of fields satisfying their fields equations modulo gauge symmetries, i.e.~the cohomology. Solutions to the field equations can be related to the original BV complex $(V, d)$ by choosing a \emph{special deformation retract}:
a decomposition  of fields
\[ V_i = \text{im} \,{d}_i \oplus H^i_{{d}} \oplus \text{im} \, k_{i+1} \;,\]
where $k$ is the so-called homotopy operator with properties listed below. One possibility is to choose a suitable Hodge star on $V$, them $\text{im} k = \text{im} \star d \star$.

\paragraph{Special deformation retracts} 
A homotopy equivalence of two complexes $(V,d_V)$ and $(W,d_W)$ is given by the data:
\begin{equation}
\begin{tikzcd}[column sep=large, row sep=small,every label/.append style={font=\normalsize}] 
 \arrow[l,"k", loop left] (V,d_V)  \arrow[r,shift left, "p"]   &(W, d_W) \arrow[l, shift left, "e" ]
\end{tikzcd}
\label{hom_eq}
\end{equation}
where the composition of \emph{embedding} $e$ and \emph{projection} $p$ satisfies:
\[
e\circ p=\hat{1}_V-d_Vk-kd_V.
\] 
This homotopy equivalence is called a deformation retract if the composition of the maps $e$ and $p$ in the reverse order comes down to an identity in the vector space $W$:
\[
p\circ e =\hat{1}_W
.\]  
There is also a special case of deformation retract called special deformation retract (or strong deformation retract), which will be of interest to us. The special deformation retract has additional conditions:
\begin{align*} 
p\circ k=0,\\
k\circ e=0,\\
k\circ k=0.
\end{align*}
If $W$ is the cohomology of $(V, d_V)$, with the trivial differential, we speak of special deformation retracts onto cohomology, or \emph{minimal models}. 

The complex $V$ might come equipped with a bilinear shifted symplectic pairing $\langle- , -\rangle$, with the differential being skew self-adjoint, i.e. \begin{equation}\label{eq:domegacompat}
\langle d_Vv, w\rangle + (-1)^{|v|}\langle v, d_Vw \rangle = 0,
\end{equation} and we could assume that $k$ is self-adjoint and $p^\dagger = e$. In this case, then, the special deformation retract is completely and equivalently specified by $I \equiv \text{im} k\subset V$; the only condition $I$ has to satisfy is that the pairing $\langle- , d_V-\rangle$ is non-degenerate on $I$ \cite[Ch.~5~Sec.~2.7]{Costello:2011but}\cite[Prop.~3.15]{Jurco:2024gct}. These SDR can be called \emph{symplectic} (used in \cite{Jurco:2024gct}) or \emph{cyclic}, which mirrors the terminology used of $L_\infty$ algebras and was suggested to us by E. Getzler. Maximal $I$ (i.e. not included in a bigger non-degenerate isotrope) correspond to SDRs where $W$ is the cohomology. 

The symplectic pairing is usually taken to have degree $-1$, but we will use a generalization to arbitrary degrees. The characterization of such SDR in terms of non-degenerate isotropes \cite[Prop.~3.15]{Jurco:2024gct} applies to this case as well, as one can easily check. The quadratic form $S_\text{free}$ is of degree $| \langle -, - \rangle|  +1$ in this case.

\paragraph{Special deformation retracts for $\mathfrak{su}(2)$}

Let us now explain which special deformation retracts one can choose on the Chevalley-Eilenberg complex of $\mathfrak{su}(2)$. We start with $E = \Lambda^\bullet \mathfrak{su}(2)^*[1]$ with a degree $-1$ pairing. The degree zero pairing  $\langle- , \mathrm{d}-\rangle$ is equal to the standard scalar product on $\Lambda^1 \mathfrak{su}(2)^*$ so that the basis $e^i$ is orthonormal; it vanishes on other components. Therefore, special deformation retracts are given by vector subspaces $I$ of $\Lambda^1 \mathfrak{su}(2)^* \cong \mathbb C^3$. This induces a decomposition
\[E \cong \underbrace{I}_{\subset \Lambda^1 \mathfrak{su}(2)^*} \oplus \underbrace{\mathrm{d}I}_{\subset \Lambda^2 \mathfrak{su}(2)^*} \oplus (I \oplus \mathrm{d}I)^\perp, \]
so that $p$ and $e$ are the projections and inclusions of the last factor, while $k$ is the inverse to $\mathrm{d} \colon I \to \mathrm{d}I$. We see that there are multiple possibilities on the dimension of $I$:
\begin{itemize}
    \item[$\dim I = 0$] This is the trivial SDR with $p = e = \text{id}$ and $k=0$.
    \item[$\dim I = 1$] There is $\mathbb{CP}^2$ worth of choices of $I$. The reduced space $W$ of the SDR consists of the (Euclidean) orthogonal to $I$ in degree $0$ (1-forms $\Lambda^1 \mathfrak{su}(2)^*$) and its image under $\mathrm{d}$ in $\Lambda^2 \mathfrak{su}(2)^*$ and the $0$-forms and $3$-forms also belong to the orthogonal $W = (I \oplus \mathrm d I )^\perp$.
    \item[$\dim I = 2$] There is again $\mathbb{CP}^2$ worth of choices, representing the Euclidean orthogonal to $I$.
    \item[$\dim I =3$] We need to take $I = \Lambda^1 \mathfrak{su}(2)^*$, i.e. there is only one possible choice. This corresponds to a SDR onto homology. The fact that there is a unique SDR in this case was observed by Cattaneo and Mn\"ev \cite[Example~2]{Cattaneo:2009deo}. $k$ explicitly reads:
    \[
    e^1e^2 \xrightarrow{k} e^3, \quad e^2e^3 \xrightarrow{k} e^1, \quad e^3 e^1 \xrightarrow{k} e^2\;.
    \]
\end{itemize}

Cattaneo and Mn\"{e}v \cite{Cattaneo:2009deo} also discuss homotopy transfer along deformation retracts
where the condition $k^2=0$ is relaxed; the other side conditions still hold.
In this $\mathfrak{su}(2)$ case, such relaxation is not possible: By the 
side conditions, we get that $d$ commutes with $k^2$ and that 
\[ k = kdk + dk^2 = kdk + k^2 d. \]
In the Chevalley-Eilenberg complex, the term $k^2 d = d k^2$ is zero, and therefore $k=kdk$ holds.
This implies that the components $\Lambda^1 \to \Lambda^0$ and $\Lambda^3 \to 
\Lambda^2$ of the homotopy are zero; since the corresponding components of $\mathrm d$ are zero. This implies that $k^2 =0$ automatically.

\paragraph{Special deformation retracts for $\mathfrak{su}(n)$}
As a consequence of working with a reducible Lie algebra, $\Lambda^\bullet \mathfrak{s}^*$ decomposes as
\begin{equation}
    \label{eq:reductive}
    \Lambda \mathfrak{s}^* = (\Lambda \mathfrak{s}^*)^\mathfrak{s} \oplus \mathfrak{s}. \Lambda \mathfrak{s}^*\;.
\end{equation}  
In view of the fact that the cohomology is given by the invariants (\cref{HKS} for $\mathfrak{su}(3)$), we want to construct an SDR 
\begin{equation}
\begin{tikzcd}[column sep=large, row sep=small,every label/.append style={font=\normalsize}] 
 \arrow[l,"k", loop left] (\Lambda \mathfrak{s}^*, \mathrm{d})  \arrow[r,shift left, "p"]   &((\Lambda \mathfrak{s}^*)^\mathfrak{s}, 0) \arrow[l, shift left, "e" ]\;,
\end{tikzcd}
\label{sdr-ce}
\end{equation}
with $e$ and $p$ being the inclusion and projection along the decomposition \eqref{eq:reductive}. Note that then $e$ is an algebra morphism in this case.

We will use the homotopy operator from \cite[Prop.~10.9]{Meinrenken}. Let us recall its construction. In addition to $\mathrm{d}$, there is a second order differential operator $\partial$ acting on $\Lambda \mathfrak{s}^*$, given by $\partial (\alpha \wedge \beta) = [\alpha, \beta]$ and $\partial \alpha = 0$ for $\alpha, \beta \in \mathfrak {s}^*$. This  bracket on $\mathfrak{s}^*$ is transferred from the Lie bracket on $\mathfrak{s}$ using a non-degenerate invariant symmetric pairing (e.g. the Killing form for semisimple $\mathfrak{s}$). The decomposition \eqref{eq:reductive} can be refined further
to 
\begin{equation} \label{eq:decrefined}\Lambda \mathfrak{s}^* = (\Lambda \mathfrak{s}^*)^\mathfrak{s} \oplus \operatorname{im}{\mathrm{d}} \oplus \operatorname{im}{\partial}\end{equation}
with $\mathrm{d}$ and $\partial$ being isomorphisms between the last two factors. The homotopy operator is defined using  $D \equiv \mathrm{d} \,\partial + \partial \, \mathrm{d}$: let $k_D$ be the homotopy inverse to $D$, i.e. inverse to $D$ on $(\Lambda \mathfrak{s}^*)^\mathfrak{s}$. Then the canonical homotopy $k$ is 
\begin{equation}
k = k_{D} \circ \partial.
\end{equation}
We now need to check that we indeed get a cyclic special deformation retract. 
\begin{proposition}
    The special deformation retract of \cite[Prop.~10.9]{Meinrenken} is cyclic with respect to the \emph{integral} pairing $\langle \alpha, \beta \rangle = (-1)^{|\alpha|}\int \alpha \wedge \beta$ on $\Lambda \mathfrak{s}^*$ for a  reductive Lie algebra $\mathfrak{s}$.
\end{proposition}
The pairing requires a choice of normalization/a reference top-form; the result of course does not depend on this choice. The pairing on $\Lambda s^*[1]$ is of degree $2-\dim{\mathfrak s}$.

For comparison with \cite{Meinrenken}, we note that we work with $\Lambda \mathfrak{s}^*$ instead of $\Lambda\mathfrak{s}$, using the identification via a pairing on $\mathfrak{s}$. Note that the ``orthogonality'' which Meinrenken \cite{Meinrenken} describes is with respect to a different pairing which pairs $k$-forms with $k$-forms; this pairing does not have a well-defined degree,
\begin{proof}
    First, we note that  integral $\int$ which selects a coefficient of a top-form in $\Lambda \mathfrak{s}^*$ satisfies the following two properties
    \begin{align} \int \mathrm{d} \alpha &= 0, \\  \int x\cdot \alpha &= 0 \label{eq:intonactionzero}\end{align}
    where $\alpha \in \Lambda \mathfrak{s}^*$ and $x\in \mathfrak{s}$, the dot denoting the action of $\mathfrak{s}$ on the module of forms $\Lambda \mathfrak{s}^*$. Both follow from the fact that the top-form is a non-trivial invariant cocycle; or they can be seen by transferring the claims to the forms on the group $\Omega^\bullet{(S)}$ (a left-invariant top form on $S$ is closed and also invariant, since reductive Lie algebras are unimodular).

    The first property implies that our pairing $\langle \alpha, \beta \rangle = (-1)^{|\alpha|}\int \alpha \wedge \beta$ is compatible with $\mathrm{d}$ as in \eqref{eq:domegacompat}; the left hand side gives $\int \mathrm d(\alpha \wedge \beta)=0$. Furthermore, our symplectic form $\langle- , - \rangle$ pairs the two subspaces $(\Lambda \mathfrak{s}^*)^\mathfrak{s} \oplus \mathfrak{s}. \Lambda \mathfrak{s}^*$ to zero by \eqref{eq:intonactionzero}. Both are thus symplectic and mutual orthogonals to each other.
    \smallskip

    Instead of directly proving that the homotopy operator of \cite{Meinrenken} is self-adjoint, we will proceed indirectly. We will show that the SDR is constructed from a non-degenerate isotrope $\text{im}\partial$, which is enough by \cite[Prop.~3.15]{Jurco:2024gct}. The homotopy operator of Meinrenken can be seen as the inverse to $\mathrm d\colon \text{im} \partial \to \text{im} \mathrm d$. Indeed, $D$ is equal to $\partial \mathrm d$ on $\text{im} \mathrm d$, and thus $k = (\partial \mathrm d)^{-1} \partial $ is inverse to $\mathrm d$. 

    Let us now check that $\text{im} \partial$ is isotropic. Since $\partial$ is a second order differential operator \cite[Prop. 10.6]{Meinrenken}, we have
    \[ \langle \partial \alpha, \partial \beta \rangle = \pm \int \partial \alpha\wedge \partial \beta  = \int \partial(\alpha \wedge \partial \beta) \pm \int [\alpha, \partial \beta].\]
    The first term vanishes since $\alpha \wedge \partial \beta$ would have to be a top $+1$ form. The second term vanishes since by \cite[Lemma~10.3.ii]{Meinrenken}, $[\alpha, \partial \beta]$ lies in $\mathfrak{s}\cdot \Lambda \mathfrak{s}^*$, and we conclude by \eqref{eq:intonactionzero}.
    
    Thus the decomposition \eqref{eq:decrefined} is exactly as in \cite[Prop.~3.15]{Jurco:2024gct}: the differential is compatible with the pairing, while the homotopy is the inverse to the differential, seen as a map from the isotropic subspace $\text{im}\partial$. Finally, $p$ and $e$ are projections along the orthogonal decomposition \eqref{eq:reductive}.
\end{proof}
Let us remark that $\partial$ can be defined using the Hodge star operator as $\star \mathrm{d} \star$, which was used in a related setting for $\mathfrak{su}(2)$ in \cite{Nguyen:2021rsa}.

\subsection{Observables} The observables of a classical field theory are functions on $V$, invariant under gauge symmetries and subjected to the field equations. Thus, we will consider the symmetric tensors on $V^*$, i.e. symmetric powers of $\text{Sym}^\bullet V^*$, with $V^*$ consisting of linear observables. The action functional is an example of such observable. The symmetric powers $\text{Sym}^\bullet(V^*)$ are the quotient of the tensor products of $V^*$ by the ideal
\[
< g\otimes f - (-1)^{|g||f|} f\otimes g>\, , \quad g, f \in V^*.
\]
where $|g|$ is the degree of $g$.

The symmetric algebra can be given the structure of a cdga. If $V^*=V[1]$, which we will always assume in the following, the Chevalley--Eilenberg differential can be given by 
\[
\mathrm{d}_{V[1]} = -\mathrm{d}
\]
and can be lifted to Sym, by enforcing the derivation property. To briefly illustrate how this works, let us take $(V^*, -\mathrm{d})$; then $\mathrm{d}_S$ coincide with $-\mathrm{d}$ on $\text{Sym}^1(V^*) \cong V^*$, while on $\text{Sym}^2(V^*)$:
\[
Q_S(\alpha \otimes_S \beta) :=   -\mathrm{d}\alpha \otimes_S \beta + (-1)^{\vert a \vert} \alpha \otimes_S -\mathrm{d}\beta
\]
and so on for higher symmetric powers. Note that acting with $Q_S$ does not affect the polynomial degree of a symmetric tensor.

The whole deformation retract extends to the symmetric algebra (\emph{tensor trick} for the symmetric \emph{cobar construction}):
\begin{equation}
    \begin{tikzcd}[column sep=large, row sep=small,every label/.append style={font=\normalsize}] 
\arrow[l,"K", loop left] (\text{Sym}^\bullet V^* , Q_S)  \arrow[r,shift left, "\Pi"] &\arrow[l, shift left, "\mathcal{E}" ] (\text{Sym}^\bullet(H^\bullet(V)^*), 0)\;.\end{tikzcd}
\label{tensor-trick}
\end{equation}
Here, $\mathcal{E}, \Pi$ and $K$  are extensions of the maps for the first order in the symmetric expansion, i.e. on $\text{Sym}^{q}(V^*)$ they act as : 
\begin{align}
    \mathcal{E}^{(q)} :=& \, \underbrace{p^{\text t} \otimes_S 
    p^{\text t}\otimes_S p^{\text t} \dots}_{q \text{ times }} \label{emb}\\
    \Pi^{(q)} := & \, e^{\text t} \underbrace{\otimes_S e^{\text t} \otimes_S e^{\text t} \dots}_{q \text{ times}} \label{proj}\\
    K^{(q)} := & \sum_{j=0}^{q} (\text{id})^{\otimes_S^j} \otimes_S k^\text t \, \otimes_S (p^{\text t} \circ e^{\text t})^{\otimes_S^{q-1-j}} \;. \label{homotopy}
\end{align}
The superscript $\text{t}$ stands for the transpose. Note that since we are on the dual, the maps $e$ and $p$ change roles.
The symmetric product of maps signifies that we symmetrize the maps. Note that
\[
\alpha \otimes_S \beta  = (-1)^{\vert \beta \vert \vert \alpha \vert} \beta \otimes_S \alpha\;, \quad \alpha,\beta \in V^*
\]
 is actually an antisymmetric product because the parity ($\mathbb{Z}_2$ grading) of both $\alpha$ and $\beta$ is odd ($\vert \alpha \vert = 1 = \vert \beta \vert$). Instead, for instance, the product in $\text{Sym}^2(V^*)$ of elements of form degree $2$, which are parity even, is genuinely symmetric.
\vspace{0.5cm}

\subsection{External Lie algebra and interactions} 
Finally, we need to connect these general considerations with the Chevalley-Eilenberg complex of $\mathfrak{su}(n)$. To encode non-trivial interactions for fields, we cannot consider $V = \Lambda \mathfrak{su}(n)^*[1]$ directly. This $V$ is a Batalin-Vilkovisky complex with a pairing and a differential, which encodes a free theory. There is a cubic tensor on this $V$ representing the wedge product, but this does not give a Batalin-Vilkovisky interacting action on $\text{Sym}^\bullet(V^*)$, essentially due to signs: as we can test on $\text{Sym}^2(V^*)$ elements,
\[
\wedge  \, (\alpha \otimes \beta + (-1)^{|\alpha||\beta|} \beta \otimes \alpha) = \alpha \wedge \beta + (-1)^{|\alpha||\beta|} \beta \wedge \alpha =0\;.
\]
Here we mention that in this setting one can use "contracted coordinate functions", see \cite{Jurco:2018sby}, which are the following objects:
\begin{equation}
    \text{Sym}V^* \otimes V[-1]\ni \varphi = \phi \otimes e\; .
\end{equation}
However, this will not be our point of view here.

Another alternative is the following: the wedge product can be seen as an interaction if instead of $\text{Sym}^\bullet(V^*)$, which is the free commutative algebra on $V^*$, we would consider the free Lie superalgebra on $V^*$; the mathematical structure behind this is the theory of homotopy commutative (or $C_\infty$) algebras.

Instead, we will stay in the realm of ordinary Batalin-Vilkovisky observables, which, in the interacting case, require $V$ to be a (homotopy) Lie algebra. This can be easily achieved by considering Lie algebra forms valued in another Lie algebra $\mathfrak g$, that is working with
\[V \equiv E_\mathfrak{g} = \Lambda \mathfrak{su}(n)^*[1]\otimes \mathfrak g.\]
Here we stress that $\mathfrak g$ is seen as a trivial $\mathfrak{su}(n)$ module. See however the discussion at the end in Section \ref{conclusions}. 

We refer to Section \ref{ssec:operads} on the compatibility between homotopy transfer before and after tensoring with $\mathfrak{g}$. In other words, from now on we demand to handle $\mathfrak{su}(2)$ forms valued in another Lie algebra, akin to Chern--Simons on a 3-dimensional group manifold. The case of $\mathfrak{su}(3)$ forms valued in an external $\mathfrak{g}$ will be compared to BF theory \cite{Schwarz:1978cn}.

Coming back to our original issue, we now consider another Lie algebra which is quadratic, i.e. it comes equipped with a non-degenerate ad-invariant symmetric pairing $\langle -,-\rangle_\mathfrak{g}$. The Killing form of a semisimple Lie algebra is an example. We can also set up a trivial deformation retract: 
\[\begin{tikzcd}[column sep=small, row sep=small,every label/.append style={font=\footnotesize}]
(\mathfrak{g}, 0) \arrow[r,shift left, "id"]  & \arrow[l,shift left, "id"] (\mathfrak{g}, 0)\;.
\end{tikzcd}\]
We can form the tensor product of \eqref{sdr-ce} with the above, to obtain ($E_\mathfrak{g} := (\Lambda \mathfrak{s}^* \otimes \mathfrak{g} )[1]$):
\begin{equation}
\begin{tikzcd}[column sep=large, row sep=small,every label/.append style={font=\normalsize}] 
 \arrow[l,"k\otimes 1", loop left] (E_\mathfrak{g},\mathrm{d}\otimes 1)  \arrow[r,shift left, "p \otimes 1"]   &(H_\mathrm{d}\otimes \mathfrak{g}, 0) \arrow[l, shift left, "e \otimes 1" ] \;.
\end{tikzcd}
\end{equation}
Later we will consider the Lie bracket as exemplary of a 2-product. By tensor trick we obtain in the $\mathfrak{su}(2)$ case: 
\[\begin{tikzcd}[column sep=small, row sep=small,every label/.append style={font=\normalsize}]
\arrow[l,"K", loop left]\;\;\;\;\;\;\;\;\;\;\;\;\;\;\;\;\;\;\;\;(\text{Sym}^\bullet((\Lambda \mathfrak{su}(2)^*)[2] \otimes \mathfrak{g}), Q_S{}_{(-\mathrm{d}\otimes 1)} ) \arrow[r,shift left, "\Pi"] & \arrow[l,shift left,"\mathcal{E}"] (\text{Sym}^\bullet(\mathfrak{g}[2] \oplus \mathfrak{g}[-1]), 0)\;,
\end{tikzcd}\]
where $Q_S{}_{(-\mathrm{d}\otimes 1)}$ is shorthand notation for the extension of the differential $-\mathrm{d}\otimes 1$ on $E[1]$ to the symmetric algebra by Leibniz rule. This homotopy transfer is promising: we will be cranking the homological perturbation machine on it.

To disclose the case of $\mathfrak{su}(3)$ we need some preliminaries on pairings and degrees.

\subsubsection{Cyclic structures and degrees} \label{ssec:cyclic}
For $\mathfrak s = \mathfrak {su}(n)$, we have a non-degenerate trace on the algebra of Chevalley-Eilenberg cochains of degree $-\dim \mathfrak{su}(n) = 1-n^2$ - i.e. a commutative Frobenius algebra with a shifted pairing in the language of \cite{Cattaneo:2009deo}. To arrive at Batalin-Vilkovisky formalism conventions, we will shift\footnote{There are two competing conventions for homotopy algebras in general. In one of them, the operations with $n$ inputs have degree $2-n$, which means that the binary operation has the familiar degree $0$. Another convention is obtained by shifting this vector space down by one degree (denoted $[1]$), where all operations have degree $1$ and, due to Koszul signs involved in these shifts, they have different symmetry properties. For example, in the first convention, $L_\infty$ operations are graded-antisymmetric, whereas in the second convention they are graded symmetric. For cyclic homotopy algebras, we need a suitable pairing; it is of degree $-3$ in the first convention and of degree $-1$ in the second convention. The second convention is more geometric, in the sense that $V[1]$ is a formal dg ($-1$ shifted symplectic) space, with the differential (or its Hamiltonian) encoding the (cyclic) homotopy algebra structure.} to 
\[ E = \Lambda^\bullet \mathfrak{su}(n)^*[1] \]
which has a pairing of degree $3-n^2$. To get a pairing of degree $-1$ on $E\otimes \mathfrak g$, the Lie algebra must have a pairing of degree $n^2-4$. Thus, for $\mathfrak{su}(n)$, we tensor with the usual metric Lie algebra; for $\mathfrak{su}(3)$ the algebra $\mathfrak g$ should have a pairing of degree $5$.

The space of fields of this theory $\Lambda^\bullet \mathfrak{su}(n)^*[1] \otimes \mathfrak g\cong \Lambda^\bullet \mathfrak{su}(n)^* \otimes \mathfrak g[1]$, should be understood as a (finite-dimensional model of the) AKSZ theory with $T[1]SU(n)$ as the source and $\mathfrak{g}[1]$ as the target. The source has a volume element of degree $1-n^2$, while the target is dg symplectic of degree $n^2-2$. 

\paragraph{$\mathfrak{su}(3)$} Concerning $\mathfrak{su}(3)$, the invariant pairing has degree
\[
\vert \langle -,-\rangle_{\mathfrak{g}}\vert =  5\; . 
\]
Then, this combines with the degree of the pairing on forms of $\Lambda \mathfrak{su}(3)^*[1]$,
\[
\left\vert \int -\wedge -\;\right\vert = -(3)^2 +3 = -6 \; 
\]
to yield a BV pairing $\langle-,-\rangle : E_\mathfrak{g} \times E_{\mathfrak{g}} \to \mathbb{C}[-1]$:
\[
\langle -, - \rangle := \int \langle - \overset{\wedge}{,} - \rangle_{\mathfrak{g}[1]} \; .
\]
Then, with the present choices of degrees for the pairings, $E_\mathfrak{g}^*\cong E_\mathfrak{g}[1]$. Now we are in the position to provide an interpretation of the BV complex. We believe it can be understood as a particular kind of AKSZ theory, a BF theory \cite{Schwarz:1978cn,Mnev:2017oko}. Indeed, notice that because of the degree of the pairing, graded components of the auxiliary Lie algebra $\mathfrak g$ of different degrees are paired: the total degree has to be $-5$. For example, we could assume that   
\[
\mathfrak{g}= \mathfrak{g}_0 \oplus \mathfrak{g}_{-5}[5]\; , \quad \langle-, - \rangle \text{ induces a non-degenerate pairing } \mathfrak{g}_0 \times \mathfrak{g}_{-5} \to \mathbb C\;.
\]
Therefore, due to our choice, $E_\mathfrak{g}$  is actually concentrated in degrees\footnote{This implies that $\mathfrak g_0$ is a non-graded Lie algebra and $\mathfrak g_{-5}\cong \mathfrak g_0^*$, with the bracket $[\mathfrak g_0, \mathfrak g_{-5}]$ given by the coadjoint action.} $0$ and $-5$
\[
E_\mathfrak{g_0} \oplus E_\mathfrak{g_{-5}}[5]\;,
\]
so if we use the BF-inspired coordinates $A = \sum A_a e^a$ on $E_{\mathfrak{g}_0}$ and $B=\sum B_a e^a$ on $E_{\mathfrak{g_{-5}}}[5]$, then:
\[
\delta A \wedge \delta B 
\]
is the degree-shifted BV symplectic form on the space of maps between the source and the target, and the reducible gauge symmetries reach ghost number $8-2=6$, as shown in the table \ref{tab:degrees}:

\begin{table}[h!]
\small
    \centering
    \begin{tabular}{c||c c c|c c c}
        forms & $\Lambda^0 \otimes \mathfrak{g}_0$ & $\Lambda^1 \otimes \mathfrak{g}_0$ & $\Lambda^2 \otimes \mathfrak{g}_0\cdots \Lambda^8 \otimes \mathfrak{g}_0$ & $\Lambda^0 \otimes \mathfrak{g}_{-5}$ & $\Lambda^1 \otimes \mathfrak{g}_{-5}$ & $\Lambda^2 \otimes \mathfrak{g}_{-5}\cdots\Lambda^6 \otimes \mathfrak{g}_{-5} \cdots \Lambda^8 \otimes \mathfrak{g}_{-5}$ \\
        \hline
       ghost deg  & $-1$ & $0$& $1 \cdots 7$  & $-6$ & $-5$ & $-4\cdots \;\;\;\;\;\; 0 \;\;\;\;\;\; \cdots 2$\\
    \end{tabular}
    \caption{The degrees of $\mathfrak{su}(3)$-forms valued in an external Lie algebra.}
    \label{tab:degrees}
\end{table}

Eventually, for the cobar construction, removing the reference to the internal degrees of $\mathfrak{g}$, we get:
\[\begin{tikzcd}[column sep=small, row sep=small,every label/.append style={font=\normalsize}]
\arrow[l,"K", loop left]\;\;\;\;\;\;\;\;\;\;\;\;\;\;\;\;\;\;\;\;(\text{Sym}^\bullet((\Lambda \mathfrak{su}(3)^*)[2] \otimes \mathfrak{g}), Q_S{}_{(-\mathrm{d}\otimes 1)}) \arrow[r,shift left, "\Pi"] & \arrow[l,shift left,"\mathcal{E}"] (\text{Sym}^\bullet(\mathfrak{g}[-2] \oplus \mathfrak{g}[1] \oplus \mathfrak{g}[3] \oplus \mathfrak{g}[6] ), 0)\;.
\end{tikzcd}\]

Later we will drop the explicit mention in the subscript $(-\mathrm{d}\otimes 1)$ and deform $Q_S$ by  $Q_S + \delta_Q$ where  $\delta_Q$ corresponds to the Lie bracket. Note that, when dealing with $\mathfrak{su}(2)$-forms, the form did on $\mathfrak g$ not need to be shifted and the above discussion would be superfluous.

Before moving to perturbation theory, an important observation to make is that the pairing $\langle -, -\rangle$ can be lifted to $\text{Sym}(E^*_{\mathfrak{g}})$ and endows that dg algebra with a Poisson bracket, derived from that pairing \cite{Nguyen:2021rsa}, \[
(\!(-,-)\!):  \text{Sym}^iE_\mathfrak{g}^* \otimes \text{Sym}^jE_\mathfrak{g}^* \to \text{Sym}^{i+j-2}E_\mathfrak{g}^*. 
\]
Here we will not delve into the details for this pairing: it exists and is used to define the action functional, or in other words we have a cyclic $L_\infty$ structure at disposal. Another way to talk about this very same object is as a $P_0$ algebra structure \cite{Costello:2021jvx},
\[\left(\text{Sym}E_\mathfrak{g}^*, Q_S, (\!(-,-)\!) \right)\;.\]
This is thus a classical BV system, and one can set out to find solutions to the classical master equation. It is known that the differential and the Lie bracket on $E_\mathfrak{g}$ will give an action with a quadratic and a cubic term $\langle - , \mathrm{d}- \rangle + \langle - , [-,-]\rangle$, which will be a solution to a classical master equation in this Gerstenhaber algebra. Here, anyway, we want to focus on its scattering amplitudes.

\section{Scattering amplitudes}

Now that a convenient interaction has been put in place, the classical S-matrix elements for an arbitrary number of legs (i.e.~the scattering amplitudes) can be derived. In this context, one can invoke the \emph{homological perturbation lemma} to obtain these.

\paragraph{Homological perturbation lemma}
In the standard treatment 
of QFT, Wick's theorem is used to express \emph{time-ordered products} of fields in terms of \emph{normal ordered products}
, which are then turned into Feynman diagrams and used to compute scattering amplitudes. The same result (though with different diagrams) is achieved with the homological perturbation lemma \cite{Crainic:2004bxw}, or also directly from considerations in homological algebra and the BV formalism as in \cite{Gwilliam:2012jg}. Sticking with the first route, let us demonstrate the construction directly in the case of the homotopy equivalence between the cobar $(\text{Sym}V^*, Q_S)$ and its cohomology. The first step is to perturb the differential $Q_S$ in a way that $Q_S+\delta_Q$ is nilpotent and the perturbation $\delta_Q$ is of the same homological degree as $Q_S$. With reference to the discussion and notation of \cref{On shell}, HPL ensures the existence of a homotopy equivalence for the perturbed data 
\[
\begin{tikzcd}[column sep=large, row sep=small,every label/.append style={font=\normalsize}]
 \arrow[l,"K'", loop left] (\text{Sym}V^*,Q_S+\delta_Q)  \arrow[r,shift left, "\Pi'"]   &(\text{Sym}W^*, Q'_S ) \arrow[l, shift left, "\mathcal{E}'" ]\;,
\end{tikzcd}
\]
where the maps are determined by those for the original homotopy equivalence and the perturbation $\delta_Q$.

\begin{lemma}\cite{Crainic:2004bxw}
For a homotopy equivalence and a map $\delta_Q: \text{Sym}V^*\rightarrow \text{Sym}V^*$ (perturbation of $Q_S$) such that $|\delta_Q|=|Q_S|$ and $(Q_S+\delta_Q)^2=0$ with invertible $(1-\delta_Q K)$, we obtain another homotopy equivalence given by:
\begin{align*}
    K'=&\; K+K(1-\delta_Q K)^{-1}\delta_Q K\;,\\
    \Pi'=&\; \Pi+ \Pi(1-\delta_Q K)^{-1}\delta_Q K\;,\\
    \mathcal{E}'=&\;\mathcal{E}+K(1-\delta_Q K)^{-1}\delta_Q \mathcal{E}\;,\\
    Q'_{S}=&\;Q_{(W)}+\Pi(1-\delta_Q K)^{-1} \delta_Q \mathcal{E}\;.
\end{align*}
\end{lemma}
Here, the quasi-isomorphisms $\mathcal{E}$ and $\Pi$, as well as the homotopy $K$, are those defined in \cref{emb,proj,homotopy}, while $Q_{(W)}$ is notation for the differential on $\text{Sym}W^*$.

Now, if the homotopy equivalence is to the (symmetric algebra of the) cohomology, the differential $Q'_S$ reconstructs the scattering amplitudes at $n$ points: it does so by relating the on-shell legs by means of propagators.  
The perturbation of the differential $\delta_Q$ could be a collection of homotopy (Lie) algebra multiproducts, 
\begin{equation}
    \delta_Q =\sum_i\mu_i, \quad \text{ with e.g.} \;\;\mu_2 = [-, -]^\text{t}
    \label{delta_Q}
\end{equation} 
i.e.~the transpose of a Lie bracket on $V[-1]$. This explains the need to shift and dualize $V$: the transpose needs to induce a map $V^* \to \text{Sym}^2 V^*$. The expression for the differential in cohomology (where $Q_{(W)}$ from above is null) is: 
\begin{equation}
    Q'_{S} = \sum_n D'^{n}\;, \quad D'^{n} = \Pi \circ \, {\delta_Q \circ (K \circ \delta_Q)^{\otimes^n_S}} \circ \,\mathcal{E}\;.
    \label{Dprime}
\end{equation}
The components $D'^{n}$ are the transposes of the higher $L_\infty$ brackets and can be expressed by sums over tree Feynman diagrams, with the homotopy operator acting on inner edges, see \cite{marklTransferring} and \cite[Lemma 9.4.7]{Vallette-Loday}. 

Our natural choice for a deformation will be the Lie bracket on $\Lambda \mathfrak{su}(n)\otimes \mathfrak g$ constructed in the previous section. Now we are provided with all the technical tools needed. 
\subsection{Vanishing of tree-level effective action}\label{gen-vanishing}
Finally, we can combine together the above results to explain why in the effective action on cohomology all the tree-level Feynman diagrams with internal edges vanish. The only diagram which remains is the cubic vertex, which captures the obvious Lie algebra structure on cohomology of $\Lambda \mathfrak{su}(n)^* \otimes \mathfrak g$.

First, it was noted in \cite[Prop.~7]{Cattaneo:2009deo} that if the inclusion $e \colon H(\Lambda \mathfrak{su}(n)^*) \to \Lambda \mathfrak{su}(n)^*$ is an algebra morphism (or at least if the image of $e$ is closed under multiplication), then each tree with an internal edge automatically vanishes. Let us recall the argument: computing the perturbed differential using \eqref{Dprime}, the term $\delta_Q$ multiplies elements of $(\Lambda \mathfrak{su}(n)^*)[2] \otimes \mathfrak g$ in the image of $e$. If the image of $e$ is closed under multiplication, the operator $k$ annihilates this term (or any other term it possibly acts on), because of the side condition $k \circ e = 0$. 

Therefore, if we find a special deformation retract for which $e$ is an algebra morphism, all the non-trivial tree diagrams (i.e. with internal edges) vanish. The SDR of \cite{Meinrenken} we are using satisfies this condition, as explained below \cref{sdr-ce}.

\subsection{Homological perturbation theory for commutative and Lie algebras}
\label{ssec:operads}
Recall that $\Lambda^\bullet \mathfrak{s}^*$ is a commutative differential graded algebra (cdga) with a trace, or equivalently a pairing. The homotopy-invariant version of a cdga is a so-called $C_\infty$ (or homotopy commutative algebra), which is an $A_\infty$ algebra where the products $m_n$ satisfy additional commutativity relations
\begin{equation}
    \label{eq:Cinfty}
\begin{aligned} 
    m_2(a_1, a_2) - m_2(a_2, a_1) &= 0 \quad  \mid \quad \mathbf{(1,1)}  \\ 
        m_3(a_1, a_2, a_3) - m_3(a_2, a_1, a_3) + m_3(a_2, a_3, a_1) &= 0 \quad \mid \quad \mathbf{(1,2)} \\
    m_3(a_1, a_2, a_3) - m_3(a_1, a_3, a_2) + m_3(a_3, a_1, a_2) &= 0 \quad \mid \quad \mathbf{(2,1)} \\
    \dots &
\end{aligned}\end{equation}
Here the pairs $\mathbf{(p,q)}$ which index these identities denote $(p,q)$-shuffles, which are permutations preserving the order of the first $p$ and last $q$ elements (here $n=p+q$ in $m_n$). If $a_i$ have non-trivial degrees, in addition to the sign of the permutation, the obvious Koszul sign has to be considered in each term. See e.g. \cite[Def.~3]{DuboisViolettePopov2012} for a general formula and also \cite[Ch.~13.1]{Vallette-Loday}.

A cdga or a $C_\infty$ structure on a vector space $A$ can be encoded by a differential on the free Lie algebra on $A[1]^*$ and the usual $A_\infty$ homotopy transfer will induce a $C_\infty$ structure on $H(A)$ \cite[Thm.~12]{ChengGetzler}, \cite[Ch.~13.1.9]{Vallette-Loday}. Thus, in general, the cohomology of the complex $\Lambda^\bullet \mathfrak{s}^*$ will be a homotopy commutative algebra with zero differential and the product induced from the wedge product, and possibly higher products.
\medskip

To get to a more familiar setting for physics, we tensored the cdga $\Lambda^\bullet \mathfrak{s}^*$ with a Lie algebra, obtaining a dgla. Performing the homotopy transfer, we get an $L_\infty$ structure on $H(\Lambda^\bullet \mathfrak{s}^*)\otimes \mathfrak g$. This $L_\infty$ algebra is related to the above $C_\infty$ algebra on $H(\Lambda^\bullet \mathfrak{s}^*)$ as follows\footnote{We thank Pavel Mn\"ev for explaining this to us.} \cite{TurchinWillwacher, RobertNicoud}. First, one can tensor a $C_\infty$ algebra with a Lie algebra as described in \cite[Sec.~7.]{TurchinWillwacher}. The result is an $L_\infty$ algebra, and by \cite[Thm.~5.1]{RobertNicoud} it agrees with the $L_\infty$ algebra coming from the homotopy tranfer.

Let us exhibit this in the simplest case. In both the $C_
\infty$ and $L_\infty$ transfers, the multilinear operations are given by sums of trees. In the $C_\infty$ case, each tree $T$ is evaluated using the commutative product on our cdga to a multilinear operation $m_T$; whereas in the $L_\infty$ case each tree gives a multilinear operation $m_T \otimes l_T$, where $l_T$ uses the Lie algebra structure on $\mathfrak g$. When computing the transferred quadratic operation, there is a single cubic tree contributing. In this case it is obvious that the homotopy transfer commutes with the tensor product $\otimes \mathfrak g$, as both quadratic operations are given by $m_{H(\Lambda^\bullet\mathfrak s^*)}\otimes [-, -]_\mathfrak g$. 

For higher products, the formula for the tensor product of a $C_\infty$ algebra with $\mathfrak g$ decomposes a $C_\infty$ product terms indexed by trees and uses these trees to evaluate the $\mathfrak g$ factor, exactly reproducing the $L_\infty$ homotopy transfer.

Finally, we note that these homotopy are closely related to the \emph{chord} or \emph{Jacobi diagrams} which appear as amplitudes for theories of Chern-Simons type; their coefficients come from $C_\infty$-transfer while the diagrams themselves encode how to contract structure constants for arbitrary quadratic $\mathfrak g$. See e.g. \cite{Bar-Natan} or \cite[Sec.~4.7.]{Mnev-notes} and \cite{QiuZabzineOddCS} for a BV treatement.

\paragraph{Formality} A dg commutative algebra $A$ is called \emph{formal} if there is a $C_\infty$-morphism $A\to H(A)$ such that its first component (a chain map $A\to H(A)$) is a quasi-isomorphism, with $H(A)$ considered with the induced cga structure (i.e. no higher operations). Formality of (cochains on) $SU(n)$ is a special case of a result on the homotopy type of K\"{a}hler manifolds by Deligne--Griffiths--Morgan--Sullivan \cite{DGMS}.

Of course, if we perform a homotopy transfer from $A$ and we end up with no higher operations on $H(A)$, we have proven that $A$ is formal; the perturbed projections and inclusions are $C_\infty$ quasi-isomorphisms. On the other hand, however, the knowledge that an algebra is formal does not guarantee that, when doing a homotopy transfer, we get $H(A)$ with no higher operations. In principle, all possible choices of the special deformation retract could give non-trivial higher operations, all of them $C_\infty$ isomorphic to $H(A)$ with no higher operations. It would be interesting to find an example where this happens. See also \cite[Remark~1]{markl2021} for a related discussion.
\smallskip

We note that the requiring that there exists an algebra inclusion $H(A) \to A$, which we exploited, is much stronger than formality.

\subsection{Other constructive proofs}

We learnt that the theory only admits the interaction vertices. Pretending we were not aware of this fact, we would still like to expand a bit on the HPL diagrams and show that they are null. 

First of all, let us be more explicit about the tensor trick and deformation. We will be lax regarding the signs and the numerical factors: after all, our considerations do not require a firm handle neither on the signs nor on the factors. Recall that one defines the perturbation coming from the Lie bracket as follows. For $\alpha \otimes x \in E^* \otimes \mathfrak g^*$ and $\varsigma_1 \otimes t_1, \varsigma_2 \otimes t_2 \in E \otimes \mathfrak{g}$
\begin{equation} \label{eq:duality}
   \left\langle \mu_2(\alpha \otimes x) \mid \varsigma_1 \otimes t_1 \otimes_S \varsigma_2 \otimes t_2 \right \rangle:= (-1)^{|\varsigma_1\otimes t_1|+|t_1|(|\varsigma_2|+1)} \left\langle \alpha \otimes x \mid \varsigma_1 \wedge \varsigma_2 \otimes \;[t_1,t_2]\right \rangle
\end{equation}
and extend to $\delta_Q$ on the symmetric algebra by the Leibniz rule. One should check that $(Q_S+\delta_Q)^2=0$, this follows by repeatedly using \eqref{eq:duality} and reduces to the compatibility of the Lie bracket and the differential on $E_\mathfrak g [-1]$ (for $[Q_S, \delta_Q]=0$) and Jacobi identity (for $\delta_Q^2=0$). 

More generally, the pairing $\text{Sym}E_\mathfrak g^*$ with $\text{Sym} E_\mathfrak g$ allows us to transport differentials on the left to \emph{codifferentials} on the right, similarly to \eqref{eq:duality}. Doing this for the pertubed differential on $\text{Sym} W^*$ in \eqref{Dprime}, one obtains a practical formula for the higher $L_\infty$ brackets of $W$ in terms of trees. These trees have leaves decorated by $e$, internal vertices by the Lie bracket, internal edges by the homotopy operator $k$ and the root by the projection; see e.g. \cite{marklTransferring} and \cite[Ch.~10.3]{Vallette-Loday}.

\paragraph{$\mathfrak{su}(2)$} It is worth, for a start, to notice that in the symmetrizer of the cohomology many elements actually vanish:
\[
\text{Sym}^n(\mathfrak{g}[2] \oplus \mathfrak{g}[-1]) = \bigoplus_k\text{Sym}^{n-k} \mathfrak{g}[2] \otimes_S \Lambda^{k} \mathfrak{g}[-1]\;.
\]
Note that here we are using the fact that symmetric powers of parity odd elements are actually antisymmetric powers. Concerning the parity, all $k$'s will work. However, if the dimension of $\mathfrak{g}$ is $d$, then when $k> d$ the forms are zero.

If $m_{n+2}$ is defined in the expression (obtained by expanding a geometric series to $n$-th order) for the $n$-th component of the new differential on the cohomology (with appropriate $\Pi$, $K$ and $\mathcal{E}$):
\begin{equation}
    D'^{(n)} = \Pi \circ \, \underbrace{\delta_Q \circ (K \circ \delta_Q)^{\otimes^n_S}}_{=:m_{n+2}} \circ \,\mathcal{E}
\end{equation}
then only $m_2, m_5, m_8 ... $ are non-zero, because only jumps by 3 are possible, as there is no cohomology in the degrees in between. Of course, the interaction vertex is given by $m_2$, while other tree level functions start from $m_5$ (6-point function).

We obtain the following result: 
\begin{proposition}\label{only-punctual-int}
The BV 
theory $(\textnormal{Sym}^\bullet(\Lambda \mathfrak{su}(2)^*[1]\otimes \mathfrak{g})^*, Q_S+\delta_Q, (\!(-,- )\!))$ admits only the tree-level three-point function $\mu_2$. 
\end{proposition}

\begin{proof} A $n$-point scattering amplitude is a process that sends $n-1$ fields into a final state, $\phi_1\dots \phi_{n-1} \mapsto \phi_n$. These fields are originally on-shell, i.e.~in cohomology, which is non-trivial only on the $0$- or the $3$-forms. Noting this, the proof is immediate: \emph{Regardless of the outcome of multiplication}, the absence of propagators $k$ between $3$-forms and $2$-forms in this theory (the $0$-form has no propagator on its own) gives us a zero result. This picture of the SDR on $E_\mathfrak g$ should illustrate the proof. 

   \begin{equation*}
        \begin{matrix}
            {\color{orange}{\Lambda^0 \mathfrak{su}_2^*  \otimes \mathfrak{g} }} &\underset{\xrightarrow{\mathrm{d}=0}}{{{\xleftarrow{k=0}}}} & {\Lambda^1 \mathfrak{su}_2^*  \otimes \mathfrak{g}} &\underset{\xrightarrow{\mathrm{d}}}{{{\xleftarrow{k}}}} & \Lambda^2 \mathfrak{su}_2^*  \otimes \mathfrak{g}& \underset{\xrightarrow{\mathrm{d}=0}}{{\color{MidnightBlue}{\xleftarrow{k=0}}}} & {\color{MidnightBlue}{\Lambda^3 \mathfrak{su}_2^*  \otimes \mathfrak{g}}} \\ 
           {\int \ast  = \pi}  \downarrow \, {\color{orange}{\uparrow e}}& &{{0  \downarrow}} \, \uparrow e & &  0 \downarrow \, \uparrow e& &{\int  = \pi}  \downarrow \, {\color{MidnightBlue}{\uparrow e}}\\
            {\color{orange}{\mathbb{C}[1] \otimes \mathfrak{g}}} & & {{\emptyset \otimes \mathfrak{g}}} & &\emptyset \otimes \mathfrak{g} & & {\color{MidnightBlue}{\mathbb{C}[-2]\otimes \mathfrak{g}}}
        \end{matrix}
        \end{equation*}
It is intended that all maps are tensored by id, so they act trivially on the dual Lie algebra. Since we are forced to start from the cohomology, non-trivial in degree $0$ and $3$, and subsequently multiply/take the Lie bracket, the outcome is that we always end up in the trivial, null situation. The homotopy $k$ is simply zero in there. This in turn implies that also applying the homotopy $K$ for the tensor trick sends us to zero. The conclusion is that only $m_2=\mu_2$, the interaction vertex, is non-zero. 
\end{proof}

The $\mathfrak{su}(2)$ field theory has only the diagrams:
\begin{equation}
\begin{tikzpicture}[baseline={([yshift=-.5ex]current bounding box.center)}]
            \matrix (m) [matrix of nodes, ampersand replacement=\&, column sep = 0.15cm, row sep = 0.2cm]{
                $\phi^{[0]}$ \& {}  \& $\phi^{[0]}$ 
                {} \& {} \& {} \\
                {} \& {} \& {} \\
                {} \& {} \& {} \\
                {} \& $\phi^{[0]}$ \& {} \\
            };
            \draw (m-3-2.center) -- (m-1-1) ;
            \draw (m-3-2.center) -- (m-1-3) ;
            \draw (m-3-2.center) -- (m-4-2) ;
        \end{tikzpicture} ~,~ 
        \begin{tikzpicture}[baseline={([yshift=-.5ex]current bounding box.center)}]
            \matrix (m) [matrix of nodes, ampersand replacement=\&, column sep = 0.15cm, row sep = 0.2cm]{
                $\phi^{[0]}$ \& {}  \& $\phi^{[3]}$ 
                {} \& {} \& {} \\
                {} \& {} \& {} \\
                {} \& {} \& {} \\
                {} \& $\phi^{[3]}$ \& {} \\
            };
            \draw (m-3-2.center) -- (m-1-1) ;
            \draw (m-3-2.center) -- (m-1-3) ;
            \draw (m-3-2.center) -- (m-4-2) ;
        \end{tikzpicture} \,.
\end{equation}

The above diagrams are easily recognizable. In Chern--Simons theory, using the familiar notation where $A\in (\Lambda^1\mathfrak{su}(2)^* \otimes \mathfrak{g})[1] $ and $C \in (\Lambda^0\mathfrak{su}(2)^* \otimes \mathfrak{g})[1]$ with the antifields carrying a superscript, they are due to the ghost-antighost field term (first summand of the Chern--Simons action)
\[
S= \langle C^*, \frac 1 2 [C,C]\rangle + {\langle A, \mathrm{d}A + \frac 1 2 [A,A]}\rangle + {\langle A^*, \mathrm{d}C + [A,C]\rangle} \, .
\]

{Before moving to further considerations, we would like to highlight that this proof neglects the outcome of multiplication while it relies only on the explicit form of the propagator.

\subsection{Partial transfer and non-trivial homotopy algebras}

So far, we have discussed homotopy transfer onto cohomology. We can also transfer
onto a bigger space $W$, in this case we can obtain non-trivial higher operations, i.e.
non-trivial scattering. 

 The pairing $\langle -, \mathrm d -\rangle$ on $E=(\Lambda^\bullet \mathfrak{su}(2)^*)[1]$ is the standard scalar product on $\mathbb R^3$ on $E_0$ and zero otherwise. A special deformation retract is completely specified by a vector subspace $I \subset \Lambda^1\mathfrak{su}(2)^*$ as described in Section \ref{On shell}. Let us now describe the transferred $C_\infty$ structures.

     \paragraph {$\dim I = 0$:} the reduction is $E$ itself and the transferred algebra is equal to the original one
     \paragraph{$\dim I = 3$:} we transfer onto homology as described above.
     \paragraph{$\dim I = 2$:} we transfer onto a space \[ \mathbb C[1] \oplus \mathbb C \oplus \mathbb C[-1] \oplus \mathbb C[-2],\] the higher product are all zero since they involve products of 1-forms which are all proportional to each other.
     \paragraph{$\dim I = 1$:} The most interesting case is $\dim I = 1$, we transfer onto the space  \[ W = \mathbb C[1] \oplus \underbrace{\mathbb C^2 \oplus \mathbb C^2[-1]}_{I^\perp} \oplus \, \mathbb C[-2]\] where the plane $W_0$ is orthogonal to $I$ in $\mathbb C^3$. We can further  identify the degree 1 component of $E$ with $\mathbb C^3$ using $\mathrm d$, i.e. using the basis $e^2e^3, e^3e^1, e^1e^2$. Using these identifications, the wedge product becomes the cross product of (complex) three-dimensional vectors, and the homotopy operator $k$ is the identity; $p$ and $i$ are orthogonal projections onto the planes $I^\perp$. 

    In particular, for $I = \mathbb C\langle e^3 \rangle$, we have
     \[
     k: e_1e_2 \mapsto e_3
     \]
     so the projector $p$ acts on $\Lambda^1$ and $\Lambda^2$ by forgetting the above terms.

    Performing homotopy transfer gives us a non-trivial triple product of vectors in $W_0$: if $a, b, c \in R_0\cong I^\perp \subset \mathbb C^3$, there are two non-trivial trees giving a  triple product
    \[ \tilde{m}_3(a,b,c) = (a\times b) \times c + a\times (b\times c) \]
    seen as a vector in $W_1$ orthogonal to $I$. This is not zero, for example again for $I = \mathbb C\langle e^3 \rangle$
    \[ \tilde{m}_3(e^1, e^2, e^1) = 2 e^2. \]
    Other possibilities of evaluating the triple product (i.e. with inputs from other degrees) are zero.\footnote{We recommend \cite{marklTransferring} for explicit signs for homotopy transfer; see also the formula \cite[Lemma 9.4.1]{Vallette-Loday} but there a Koszul sign is missing.} The reader is invited to check that this satisfies  the $C_\infty$ symmetry \eqref{eq:Cinfty}. 

   There are non-trivial binary products given by the diagram with a single vertex, we only note that the binary product vanishes on one-forms in $I^\perp$, i.e. $W_0$. 
    
    Finally, let us explain why quartic and higher products vanish. Let us note that
    \begin{align*}
         I\times I &= 0  \\
        I^\perp \times I &\subset I^\perp \\
        I^\perp \times I^\perp &\subset I
    \end{align*}
   We start with $I^\perp$ seen as $1$-forms. Binary vertices also give us one-forms in $I$, by the last line. Since $I\times I = 0$, this $I$ can be combined with $I^\perp$ only, giving $I^\perp$. This has to be followed by the projection, as the homotopy is zero on (two-forms in) $I^\perp$. This is our triple product. 
   \begin{equation}
\begin{tikzpicture}[baseline={([yshift=-.5ex]current bounding box.center)}]
            \matrix (m) [matrix of nodes, ampersand replacement=\&, column sep = 0.15cm, row sep = 0.2cm]{
                {} \& $I^\perp$ \& $I^\perp$   \& $I^\perp$ 
                {} \& {} \& {} \& {}  \\
                {} \& {} \& {} \& {} \\
                {} \& {} \& {} \& {} \\
                $I^\perp$ \& {} \& {}  \& {}\\
            };
            \draw (m-1-4) -- (m-2-3.center) ;
            \draw (m-1-3) -- (m-2-3.center) ;
            \draw[ambi] (m-2-3.center) -- (m-3-2.center) ;
            \draw (m-1-2) -- (m-3-2.center) ;
            \draw (m-4-1) -- (m-3-2.center) ;
        \end{tikzpicture}
\end{equation}
The diagram illustrates how the only way to obtain a non-zero result is by feeding $I^\perp$ to an internal propagator $k$ (dashed line), and then projecting on-shell.
\smallskip

A similar case-by-case inductive reasoning shows that there are no higher products involving the unit, the top-form or two-forms.

\paragraph{$\mathfrak{su}(3)$}

Note that in this case, all $m_n$ are in principle allowed: we have no numerical considerations to impose to look at just some $m_n$'s with fixed $n$.

Let us proceed to convince ourselves that the diagrams are null. We consider, in \cref{table-modules}, the modules with weights $(l,p)$ in the decomposition of the exterior algebra for $\mathfrak{sl}(3)$, the complexification of $\mathfrak{su}(3)$. The singlet $V(0,0)$ describes cohomology, while the rest are other irreducible spaces.

\begin{table}[h!]
    \centering
    \begin{tabular}{c|c}
        $\Lambda^4$ & $\{V(1,1) , V(2,2)\}$ multiplicity $2$\\
        \hline
        $\Lambda^3$, $\Lambda^5$ & $V(0,0)$, $V(1,1)$, $V(2,2)$, $V(3,0)$, $V(0,3)$\\
        \hline
        $\Lambda^2$, $\Lambda^6$ & $V(1,1)$, $V(3,0)$, $V(0,3)$\\
        \hline
        $\Lambda^1$, $\Lambda^7$ & $V(1,1)$\\
        \hline
        $\Lambda^0$, $\Lambda^8$ & $V(0,0)$\\
    \end{tabular}
    \caption{Modules labeled by their weights.}
    \label{table-modules}
\end{table}

According to \cref{table-modules}, we are guaranteed the existence of a non-zero homotopy for the Chevalley--Eilenberg differential $\mathrm{d}$ between all $j$ and $j-1$ forms, other than between $8$ and $7$ or between $1$ and $0$. This sounds promising, however it will not be sufficient.

\noindent All seemingly plausible diagrams with an internal propagator go to zero. The reason is that the product involves the singlet, while the propagator can only be fed elements in the complement to the cohomology. In other words, an application of $\delta_Q$ multiplies two singlets $V(0,0)$ while the propagator involves all the other modules, so when these are combined, the end result is zero. 
Pictorially, 
\begin{equation}
        \begin{tikzpicture}[baseline={([yshift=-.5ex]current bounding box.center)}]
            \matrix (m) [matrix of nodes, ampersand replacement=\&, column sep = 0.15cm, row sep = 0.2cm]{
                {} \& $\phi^{[0]}$ \& $\phi^{[0]}$  \& $\phi^{[3]}$ \\
                {} \& {} \& {} \& {}  \& {} \\
                {} \& {} \& {} \& {}  \& {} \\              
            };
            \draw (m-2-3.center) -- (m-1-4) ;
            \draw (m-2-3.center) -- (m-1-3) ;
            \draw[int] (m-3-2.center) -- (m-2-3.center) ;
            \draw (m-3-2.center) -- (m-1-2) ;
        \end{tikzpicture} \;.
        \end{equation}
The dashed line symbolises the absence of a non-zero propagator between $V(0,0)$ and a module for $\Lambda^2$. To sum up, we have proven a version of \ref{only-punctual-int} for $\mathfrak{su}(3)$, recovering the argument about the embedding of invariants from representation theory of $\mathfrak{su}(3)$.
\begin{proposition}
The BV 
theory $(\text{Sym}^\bullet(\Lambda \mathfrak{su}(3)^*\otimes \mathfrak{g})^*, Q_S+\delta_Q, (\!(-,- )\!))$ admits only the tree-level three-point function $\mu_2$.
\end{proposition}

The $\mathfrak{su}(3)$ field theory has only the diagrams:
\begin{align}
\begin{tikzpicture}[baseline={([yshift=-.5ex]current bounding box.center)}]
            \matrix (m) [matrix of nodes, ampersand replacement=\&, column sep = 0.15cm, row sep = 0.2cm]{
                $\phi^{[0]}$ \& {}  \& $\phi^{[0]}$ 
                {} \& {} \& {} \\
                {} \& {} \& {} \\
                {} \& {} \& {} \\
                {} \& $\phi^{[0]}$ \& {} \\
            };
            \draw (m-3-2.center) -- (m-1-1) ;
            \draw (m-3-2.center) -- (m-1-3) ;
            \draw (m-3-2.center) -- (m-4-2) ;
        \end{tikzpicture} ~,~ 
        \begin{tikzpicture}[baseline={([yshift=-.5ex]current bounding box.center)}]
            \matrix (m) [matrix of nodes, ampersand replacement=\&, column sep = 0.15cm, row sep = 0.2cm]{
                $\phi^{[0]}$ \& {}  \& $\phi^{[3]}$ 
                {} \& {} \& {} \\
                {} \& {} \& {} \\
                {} \& {} \& {} \\
                {} \& $\phi^{[3]}$ \& {} \\
            };
            \draw (m-3-2.center) -- (m-1-1) ;
            \draw (m-3-2.center) -- (m-1-3) ;
            \draw (m-3-2.center) -- (m-4-2) ;
        \end{tikzpicture} 
        ~,~ 
        \begin{tikzpicture}[baseline={([yshift=-.5ex]current bounding box.center)}]
            \matrix (m) [matrix of nodes, ampersand replacement=\&, column sep = 0.15cm, row sep = 0.2cm]{
                $\phi^{[0]}$ \& {}  \& $\phi^{[5]}$ 
                {} \& {} \& {} \\
                {} \& {} \& {} \\
                {} \& {} \& {} \\
                {} \& $\phi^{[5]}$ \& {} \\
            };
            \draw (m-3-2.center) -- (m-1-1) ;
            \draw (m-3-2.center) -- (m-1-3) ;
            \draw (m-3-2.center) -- (m-4-2) ;
        \end{tikzpicture}
        ~,~ \notag \\
        \begin{tikzpicture}[baseline={([yshift=-.5ex]current bounding box.center)}]
            \matrix (m) [matrix of nodes, ampersand replacement=\&, column sep = 0.15cm, row sep = 0.2cm]{
                $\phi^{[0]}$ \& {}  \& $\phi^{[8]}$ 
                {} \& {} \& {} \\
                {} \& {} \& {} \\
                {} \& {} \& {} \\
                {} \& $\phi^{[8]}$ \& {} \\
            };
            \draw (m-3-2.center) -- (m-1-1) ;
            \draw (m-3-2.center) -- (m-1-3) ;
            \draw (m-3-2.center) -- (m-4-2) ;
        \end{tikzpicture}
        ~,~ 
        \begin{tikzpicture}[baseline={([yshift=-.5ex]current bounding box.center)}]
            \matrix (m) [matrix of nodes, ampersand replacement=\&, column sep = 0.15cm, row sep = 0.2cm]{
                $\phi^{[3]}$ \& {}  \& $\phi^{[5]}$ 
                {} \& {} \& {} \\
                {} \& {} \& {} \\
                {} \& {} \& {} \\
                {} \& $\phi^{[8]}$ \& {} \\
            };
            \draw (m-3-2.center) -- (m-1-1) ;
            \draw (m-3-2.center) -- (m-1-3) ;
            \draw (m-3-2.center) -- (m-4-2) ;
        \end{tikzpicture}
        \,.
\end{align}

\subsection{Considerations about the loop level}
Let us finish with some comments about the quantum corrections to the calculations described above. First, it was noted in \cite{Cattaneo:2009deo} that in order for the action $S$ to satisfy the quantum master equation
\[ \hbar \Delta S + \frac 12 \{ S, S\} = 0 \;,\]
the Lie algebra $\mathfrak g$ needs to be unimodular (the operator $\text{ad}_x \colon \mathfrak g \to \mathfrak g$ has zero trace for each $x\in \mathfrak g$). Let us supplement this observation with noting that if $\mathfrak g$ is not unimodular, it is not possible to correct $S$ by $\hbar$-dependent corrections to a quantum master action. 
To argue this, let us choose a basis of the dual of $\Lambda \mathfrak{su}(2)^*\otimes \mathfrak g [1]$ given as 
\begin{align*}
    u_\alpha &\text{ basis of } \Lambda^0; \text{ degree }1 \\
    x^i_\alpha &\text{ basis of } \Lambda^1; \text{ degree }0\\
    y^i_\alpha &\text{ basis of } \Lambda^2; \text{ degree }-1\\
    v_\alpha &\text{ basis of } \Lambda^3; \text{ degree }-2
\end{align*}
Here $i=1,2,3$ is the $\mathfrak{su}(2)$ index, while $\alpha$ is the $\mathfrak{g}$ index; the basis vectors $x^i$ are dual to $e^i$ and $y^i$ dual to $\tfrac 12 \epsilon_{ijk}e^je^k$. The classical BV action is
\begin{equation}
    \label{eq:classicalBVaction}
    S = t^{\alpha\beta} x^i_\alpha x^i_\beta + f^{\alpha \beta \gamma}\left(\frac{1}{3!}x^i_\alpha x^j_\beta x^k_\gamma \epsilon_{ijk}+ u_\alpha x^i_\beta y^i_\gamma +\frac{1}{2}u_\alpha u_\beta v_\gamma\right)
\end{equation}
while the BV operator is
\[ \Delta = t_{\alpha\beta}\sum_i \frac{\partial^2}{\partial x^i_\alpha \partial y^i_\beta} + t_{\alpha\beta} \frac{\partial^2}{\partial u_\alpha \partial v_\beta} \]
Here, $t^{\alpha\beta}$ are the matrix coefficients of the pairing on $\mathfrak g$, while $f^{\alpha\beta\gamma}$ are the (appropriately lifted) structure constants of $\mathfrak g$. 

Computing $\Delta S$, we get a term proportional to 
\[ f^{\alpha \beta \gamma}t_{\alpha\beta} u_\gamma\]
which leads to the condition on unimodularity. This term would need to be counteracted by another action  $S'$ which would be proportional to $\hbar$ and linear in fields, so that \[ \{ S', t^{\alpha\beta} x^i_\alpha x^i_\beta\} + \hbar \Delta S = 0 \]
However, we cannot choose $S'$ so that the first term is proportional to $u_\alpha$ (essentially since $u_\alpha$ is nontrivial in the cohomology).  

\paragraph{Quantum corrections to effective actions.} In the $\mathfrak{su}(2)$ case (and others), Cattaneo and Mn\"{e}v have computed all quantum corrections to the effective action on cohomology, see \cite[Example~2]{Cattaneo:2009deo}. Let us sketch the quantum corrections in the case we transfer onto a bigger subspace using a one-dimensional isotrope $I = \mathbb C\langle e^3\rangle$. The propagator is 
\[ t_{\alpha\beta}\frac{\partial^2}{\partial x^3_\alpha \partial x^3_\beta}\]
and the BV fiber integral is computed by setting $y^3=0$ and contracting all $x^3$ using the propagator. Apart from the tree level diagrams, we get no higher quantum corrections.

{With $I= \mathbb{C}\langle e^1,e^2\rangle$, the propagator $t_{\alpha\beta}\frac{\partial^2}{\partial x^1_\alpha \partial x^1_\beta} + t_{\alpha\beta}\frac{\partial^2}{\partial x^2_\alpha \partial x^2_\beta}$ and with the BV integral computed by setting $y^1=0=y^2$, the only relevant term in the interacting action is $x^i_\alpha x^j_\beta x^k_\gamma \epsilon_{ijk}$, which means we have a three-valent vertex with half-edges labelled by $1, 2, 3$. The only non-trivial diagrams are wheels of even diameter, with the inner edges alternating between $1$ and $2$, e.g.:
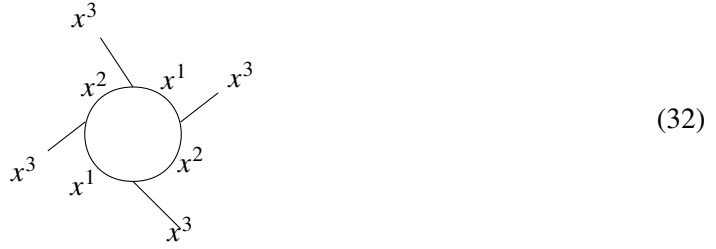
     \begin{equation}\label{diag:2-loops}
            \begin{tikzpicture}[baseline={([yshift=-.5ex]current bounding box.center)}]
                \matrix (m) [matrix of nodes, ampersand replacement=\&, column sep = 0.2cm, row sep = 0.2cm]{
                    {}\& $x^3$ \& {}  \& {} \& {}\\
                    {} \& {} \& {} \& {} \& $x^3$\\
                    {} \& {} \& {} \& {} \& {}\\
                    $x^3$ \& {} \& {} \& {} \& {}\\
                    {} \& {} \& {} \& $x^3$ \& {}\\
                };
                \draw (m-1-2) --  (m-2-3.center) ;
                \draw (m-2-5) -- (m-3-4.center) ;
                \draw (m-4-1) -- (m-3-2.center) ;
                \draw (m-5-4.center) -- (m-4-3.center) ;
                \draw (m-2-3.center) to[out=0,in=0, distance=0.85cm]  (m-4-3.center) ;
                \draw (m-2-3.center)  to[out=180,in=180, distance=0.85cm] (m-4-3.center) ;
                \node[align=center,black] at (0.55,0.6) {$x^1$};
                \node[align=center,black] at (-0.5,0.5) {$x^2$};
                \node[align=center,black] at (0.77,-0.5) {$x^2$};
                \node[align=center,black] at (-0.65,-0.8) {$x^1$};
            \end{tikzpicture}
        \end{equation}

\section{Conclusions and remarks}
\label{conclusions}
In this short note, we have investigated a toy model for a field theory with gauge symmetries. Our investigation is perturbative though we focused on the classical level, leaving considerations about the quantum corrections for the last part of this article.  Our setup is homological and in particular it is based on the Batalin--Vilkovisky formalism, which encompasses standard techniques of perturbative QFTs.
We have constructed a finite dimensional gauge theory over the $3$-sphere (group manifold associated to $\mathfrak{su}(2)$) in the classical BV formalism. We have also gone higher in the dimension, and built a finite dimensional gauge theory over $\mathbb{S}^3 \times \mathbb{S}^5$. Our Lie algebra forms take value in another Lie algebra. We showed that only the interaction vertex survives. Internal legs could exists: some propagators are non-zero. Unfortunately, by HPL (see formula \eqref{Dprime} for $D'$) the output of a $2$-product lies in cohomology, thus it cannot make up an internal leg, which instead \emph{always} lies in the complement to cohomology. That is thus null. We point out similarities to \cite{Nguyen:2021rsa}, which studies $\mathfrak{su}(2)$-forms but valued in spherical harmonics. In our notation, this means they consider $\Lambda\mathfrak{su}(2)^*[1]\otimes \mathfrak g$ with $\mathfrak g = \operatorname{End{V}}$ where $V$ is an irrep of $\mathfrak{su}(2)$. This allows them to add another term in the classical BV action, accounting for the action of $\mathfrak{su}(2)$ on $\mathfrak{g}$ i.e. they consider the Chevalley-Eilenberg complex with (Lie algebra) coefficients. In the notation of \eqref{eq:classicalBVaction}, the additional term of the action is
\begin{equation}\label{eq:ActionTerm} \epsilon^{ijk}\rho^{\alpha\beta}_{ k} \, x^i_\alpha x^j_\beta + \rho^{\alpha\beta}_{ i}\,u_\alpha y^i_\beta ,\end{equation}
where $ \rho^{\alpha\beta}_{ i}$ are the structure constants for the action of $\mathfrak{su}(2)$ on $\mathfrak{g}$. We note that the argument for vanishing of tree diagrams we use in Section \ref{gen-vanishing} also applies in their case, since the inclusion \cite[Eq.~3.61]{Nguyen:2021rsa} embeds the cohomology as an (abelian) Lie subalgebra. However, one cannot relate the results of \cite{Nguyen:2021rsa} with ours using the homological perturbation theory directly:  the additional term \eqref{eq:ActionTerm} is not ``small'' in the sense of HPL, since one cannot see the cohomology \cite[Eq.~3.60]{Nguyen:2021rsa} as coming from our cohomology by the way of an induced differential, due to degree reasons.

The message of this short article is quite simple, nevertheless because of the applied nature and the explicit formulas we hope that it could be useful to physicists working on scattering amplitudes, as well as algebraists interested in mathematical physics. A possible venue to explore for non-trivial products in cohomology are Heisenberg nilpotent groups (thanks P. Mn\"{e}v for explaining this to us).

\section*{Acknowledgments} We are grateful to Andrey Krutov for his valuable expertise in Lie theory and many helpful comments on a draft of this paper. We thank Pavol \v{S}evera, Bra\v{n}o Jur\v{c}o, Frico Valach, Domenico Fiorenza  for consultations and Pavel Mn\"{e}v and Alex Schenkel for reading and commenting on a draft of this paper. We are especially thankful to Pavel Mn\"ev for explaining to us the compatibility of $C_\infty$ and $L_\infty$ homotopy transfers. 
E.B.~acknowledges support by GA\v{C}R grant PIF-OUT 24-10634O and by COST action CaLISTA CA21109 supported by COST (European Cooperation in Science and Technology) for the opportunity to present part of this work in Corfu at the 2025 workshop. E.B.~warmly thanks the participants and the organizers for a productive meeting in a friendly environment. The work of J.P. was supported by the GAČR grant PIF-25-17640I. J.P. would like to thank M. Vorobel for discussions regarding the space of non-degenerate isotropic subspaces. E.B. and J.P. would both like to thank Institut Mittag Leffler where this work was started.

\end{document}